\documentclass[11pt]{article}
\usepackage{algorithm, algpseudocode, amsmath, amsthm, amssymb, bm, fullpage, graphicx}
\usepackage[colorlinks,pdftex]{hyperref} 
\usepackage[shortlabels]{enumitem}

\newcommand{\rC}{\mathrm{CYC}}
\newcommand{\rI}{\mathrm{INV}}

\usepackage{amsmath}

\let\bra\relax
\let\ket\relax

\DeclareMathOperator*{\E}{\mathbb{E}}
\DeclareMathOperator{\supp}{supp}
\DeclareMathOperator{\tr}{tr}
\newcommand{\abs}[1]{\left\vert#1\right\vert}

\newcommand{\bra}[1]{\left\langle#1\right|}
\newcommand{\ket}[1]{\left|#1\right\rangle}
\newcommand{\braket}[2]{\left\langle#1|#2\right\rangle}

\newcommand{\setb}[2]{\left\{#1 : #2\right\}}

\newcommand{\nth}[1]{\ensuremath{{#1}^{\mathrm{th}}}}

\let\defref\defn

\let\thmref\thm

\newcommand{\bbC}{\mathbb{C}}

\newcommand{\br}{{\bm{r}}}
\newcommand{\bx}{{\bm{x}}}
\newcommand{\by}{{\bm{y}}}
\newcommand{\bz}{{\bm{z}}}

\def\cO{{\mathcal{O}}}
\def\bbC{\mathbb{C}}
\def\ot{\otimes}
\newcommand{\proj}[1]{\left|#1\right\rangle\!\left\langle#1\right|}
\def\benum{\begin{enumerate}}
\def\eenum{\end{enumerate}}
\def\bit{\begin{itemize}}
\def\eit{\end{itemize}}

\def\bx{{\bm{x}}}
\def\by{{\bm{y}}}
\def\bz{{\bm{z}}}
\def\br{{\bm{r}}}

\newtheorem{lemma}{Lemma}
\newtheorem{definition}[lemma]{Definition}

\newtheorem{proposition}[lemma]{Proposition}
\newtheorem{corollary}[lemma]{Corollary}
\newtheorem{theorem}[lemma]{Theorem}
\newtheorem{transformation-rule}[lemma]{Transformation Rule}

\begin{document}

\title{Uselessness for an Oracle Model with Internal Randomness}
\author{Aram W. Harrow \\
{\small Center for Theoretical Physics} \\
{\small Massachusetts Institute of Technology}
 \and 
David J. Rosenbaum\thanks{Corresponding author: {\tt djr@cs.washington.edu}}\\
{\small Department of Computer Science \& Engineering} \\
{\small University of Washington}
}
\date{\today}
\maketitle
\thispagestyle{empty}

\begin{abstract}
  We consider a generalization of the standard oracle model in which the oracle acts on the target with a permutation selected according to internal random coins.  We describe several problems that are impossible to solve classically but can be solved by a quantum algorithm using a single query; we show that such infinity-vs-one separations between classical and quantum query complexities can be constructed from much weaker separations.

  We also give conditions to determine when oracle problems---either in the standard model, or in any of the generalizations we consider---cannot be solved with success probability better than random guessing would achieve.  In the oracle model with internal randomness where the goal is to gain any nonzero advantage over guessing, we prove (roughly speaking) that $k$ quantum queries are equivalent in power to $2k$ classical queries, thus extending results of Meyer and Pommersheim.
\end{abstract}

\newpage
\setcounter{page}{1}

\section{Introduction}
Oracles are an important conceptual framework for understanding quantum speedups.  They may represent subroutines whose code we cannot usefully examine, or an unknown physical system whose properties we would like to estimate.  When used by a quantum computer, the most general form of an oracle is a possibly noisy quantum operation that can be applied to an $n$-qubit input.  However, oracles this general have no obvious classical analogue, which makes it difficult to compare the ability of classical and quantum computers to efficiently interrogate oracles.  This was the original motivation of the {\em standard oracle model}, in which $f$ is a function from $[N]=\{1,\ldots,N\}$ to $\{0,1\}$, and the oracle $\mathcal{O}_f$ acts for a classical computer by mapping $x,y$ to $x,y\oplus f(x)$, and for a quantum computer as a unitary that maps $\ket{x,y}$ to $\ket{x,y\oplus f(x)}$.  One way to justify the standard oracle model is that if there is a (not necessarily reversible) classical circuit computing $f$, then $\mathcal{O}_f$ can be simulated by computing $f$, XORing the answer onto the target, and uncomputing $f$. 


In this paper, we consider other forms of oracles that are more general than the standard oracle model, but nevertheless permit comparison between classical and quantum query complexities.  Meyer and Pommersheim~\cite{meyer2009a} generalized the standard model by letting $A$ be a deterministic classical algorithm that takes the control $x$ of the oracle and computes a value $A(x)$.  The oracle then acts by applying a permutation $\pi_{A(x)}$ to the target.  We will further generalize the model by replacing $A$ with a randomized classical algorithm.  The random coins used by $A$ are internal to the oracle and cannot be accessed externally.  We call this concept an {\em oracle with internal randomness}. Note that even if $A$ takes no input, the oracle can still be interesting since it may apply different permutations depending on its internal coin flips.


Oracles with internal randomness correspond naturally to the situation in which a (quantum or classical) computer seeks to determine properties of a device that acts in a noisy or otherwise non-deterministic manner.  One simple example is an oracle that ``misfires'', i.e. when queried, the oracle does nothing with probability $p$ and responds according to the standard oracle model with probability $1-p$.  This model was considered in~\cite{regev2008a}, which found, somewhat surprisingly, that the square-root advantage of Grover search disappears (i.e. there is an $\Omega(N)$ quantum query lower bound for computing the OR function) for any constant $p>0$.

The rest of our paper is divided into two parts.  First, we explore various examples of oracles with internal randomness that demonstrate the power of the model.  We will see that in some cases (e.g. Theorems~\ref{thm:inv-unlimited} and \ref{thm:inv-quantum}), this can even result in problems solvable with one quantum query that are completely unsolvable using classical queries. 

In the second part, we consider the question of when oracle problems can be solved with any nontrivial advantage; i.e. a probability of success better than could be obtained by simply guessing the answer according to the prior distribution.  For an example of when such advantage is {\em not} possible, consider the parity function on $N$ bits.  If these bits are drawn from the uniform distribution, then any classical algorithm making $\leq N-1$ queries---or any quantum algorithm making $\leq \frac{N}{2}-1$ queries---will not be able to guess the parity with any nontrivial advantage.  In Section~\ref{uselessness}, we consider the problem of when some number of queries are {\em useless} for solving an oracle problem.  Informally, our main result is roughly that $k$ quantum queries are useless if and only if $2k$ classical queries are useless (this is formalized in Theorem~\ref{randomized-uselessness-theorem}).  However, a subtlety arises in our theorem when oracles have internal randomness, in that the $2k$ classical queries need to be considered as $k$ pairs, each of which uses a separate sample from the internal randomness of the oracle.


In the unbounded-error query complexity regime, similar results were obtained 15 years ago by Farhi, Goldstone, Gutman and Sipser~\cite{farhi1998a} for the case of the parity function.  More recently, Montanaro, Nishimura and Raymond~\cite{montanaro2007a} proved a similar result for any binary function $f$, using techniques that do not readily generalize to non-binary $f$.  One direction of the special case of our result for deterministic permutative oracles was proved by Meyer and Pommersheim~\cite{meyer2010a}.  Our proof is arguably simpler and more operational.  We introduce an analogue of gate teleportation~\cite{gottesman1999a} for oracles by showing that oracles can be (a) encoded into states analogous to Choi-Jamio\l{}kowski states, and (b) retrieved from those states with an exponentially small, but heralded, success probability (i.e. the procedure outputs a flag that tells us whether it succeeded or failed).  We expect that this characterization will be useful for future study of query complexity in the regime where any nonzero advantage is sought.

Finally, our encoding can be used to construct infinity-vs-one separations from {\em any} separation between classical and quantum uselessness (see Theorems~\ref{thm:gen-classical-uselessness} and \ref{thm:gen-quantum-uselessness}). 

\section{Examples of infinity-vs-one query-complexity separations}
\label{sec:infinity-vs-one}
In this section, we discuss problems that can be solved using a single quantum query but cannot be solved classically even with an unlimited number of queries.  Such a separation is far stronger even than exponential separations.  To achieve such infinity-vs-one separations, it is necessary (but not sufficient) for the oracle to have internal randomness since otherwise one could simulate the quantum algorithm classically with exponential overhead.  The key point is that internal randomness effectively causes a different oracle to be used for each query so such a simulation is not possible in this case.

\subsection{Distinguishing involutions with no fixed points from cycles}
\label{distinguishing-involutions-from-cycles-problem}
Our first example of an infinity-vs-one separation is given by the problem of distinguishing involutions from cycles.  Define $\rI = \setb{\pi \in S_N}{\pi^2 = 1 \text{ and } \pi x \not= x \text{ for all } x \in [N]}$; this is the set of involutions in $S_N$ with no fixed points.  Let $\rC = \setb{\pi \in S_N}{\pi \text{ is a cycle of length } N}$.  For any nonempty subset $S$ of $S_N$, define $\mathcal{O}_S$ to be the oracle with a control $x \in [N]$ and a target $y \in [N]$ that acts according to Algorithm~\ref{distinguishing-involutions-from-cycles-oracle}.

\begin{algorithm}
  \centering
  \begin{algorithmic}[1]
    \State Select $\pi \in S$ uniformly at random
    \State Compute $\pi(x)$ where $x$ is the value of the control
    \State Add $\pi(x)$ to the target $y$ modulo $N$
  \end{algorithmic} 
  \caption{The oracle for the problem of distinguishing involutions with no fixed points from cycles}
  \label{distinguishing-involutions-from-cycles-oracle}
\end{algorithm}

\begin{theorem}
\label{thm:inv-unlimited}
  Classical algorithms cannot solve the problem of distinguishing cycles from involutions with no fixed points with unbounded error using any number of queries.
\end{theorem}

\begin{proof}
In this problem, an oracle $\mathcal{O}_S$ is given which is either $\mathcal{O}_{\rI}$ or $\mathcal{O}_{\rC}$; the problem is to determine which of these is the case.  Consider querying the oracle when the control is $x$.  Then $\pi(x)$ is a uniformly random value in $[N] \setminus \{x\}$ for both cases so this problem cannot be solved by a classical algorithm.
\end{proof}

However, the problem can be solved by a quantum algorithm using a single query to the oracle as shown in Algorithm~\ref{distinguishing-involutions-from-cycles-algorithm}.

Our algorithm will make use of the swap test~\cite{buhrman2001a}.  Let $F$ denote the swap operator, i.e. $F\ket{\alpha,\beta} = \ket{\beta,\alpha}$.  Then the swap test is defined to be the measurement with outcomes $\{\frac{I+F}{2}, \frac{I-F}{2}\}$.  The $\frac{I+F}{2}$ outcome is called the ``symmetric'' outcome and $\frac{I-F}{2}$ the ``antisymmetric'' outcome.  A key property of the swap test is that when applied to pure states $\ket{\alpha}\ot \ket{\beta}$, it will output ``symmetric'' with probability $\frac{1  + |\braket{\alpha}{\beta}|^2}{2}$, thus providing an estimate of their overlap.

\begin{algorithm}
  \centering
  \begin{algorithmic}[1]
    \State Prepare the state $\frac{1}{\sqrt{N}} \sum_{x = 1}^N \ket{x}$
    \State Apply $\mathcal{O}_S$ to obtain the state $\frac{1}{\sqrt{N}} \sum_{x = 1}^N \ket{x, \pi(x)}$
    \State Apply the swap test to $\frac{1}{\sqrt{N}} \sum_{x = 1}^N \ket{x, \pi(x)}$
    \If{the swap test outputs ``symmetric''}
      \Return ``INV''
    \Else\ 
      \Return ``CYC''
    \EndIf
  \end{algorithmic}
  \caption{The quantum algorithm for distinguishing involutions with no fixed points from cycles}
  \label{distinguishing-involutions-from-cycles-algorithm}
\end{algorithm}

We now show that the above algorithm effectively counts the number of transpositions in an arbitrary permutation which is sufficient to distinguish involutions from cycles.

\begin{theorem}\label{thm:inv-quantum}
  Quantum algorithms can solve the problem of distinguishing cycles from involutions with no fixed points using a single query with one-sided error $1 / 2$.
\end{theorem}

\begin{proof}
  Consider a general state $\rho^{AB}$ on two identical systems $A$ and $B$.  Then applying the swap test to this system (where the swap exchanges $A$ and $B$) outputs $0$ (symmetric) with probability $\Pr(0) = \frac{1 + \tr{\rho^{AB} F}}{2}$.  Applying this formula to the state $\frac{1}{\sqrt{N}} \sum_{x = 1}^N \ket{x, \pi(x)}$, the probability of observing $0$ is
  
  \begin{align}
    \Pr(0) & = \frac{1 + (1 / N) \sum_{x y} \braket{\pi(y)}{x} \braket{y}{\pi(x)}}{2} \\
        {} & = \frac{1 + (1 / N) \abs{\setb{(x, y)}{\pi(x) = y \text{ and } \pi(y) = x}}}{2}
  \end{align}
  This probability is $1 / 2$ if $\pi \in \rC$ and is $1$ if $\pi \in \rI$.
\end{proof}

Hence, there is an infinity-vs-one separation in the unbounded-error classical and quantum query complexities for this problem.  This analysis can also be applied to obtain an algorithm for estimating the number of transpositions in any permutation.


\subsection{An infinity-vs-$\Theta(n)$ separation for a modification of Simon's problem}
We now show how to modify Simon's problem~\cite{simon1994a} to obtain an infinity-vs-$\Theta(n)$ separation between the classical and quantum query complexities.  Recall that for Simon's problem, we are given oracle access to a function $f : \mathbb{Z}_2^n \rightarrow \mathbb{Z}_2^n$ and $f(x) = f(y)$ if and only if $x = y + a$ for some fixed element $a \in \mathbb{Z}_2^n$ and our task is to determine $a$.  Classically, exponentially many queries are required; however, quantumly at each step we learn a vector that is orthogonal to $a$ so that the expected number of queries required is $\Theta(n)$.  The crucial point here is that this algorithm will return a vector orthogonal to $a$ for any $f$ that is constant and distinct on the cosets $\{x, x + a\}$, so if $f$ changes between calls to the oracle and $a$ does not, then the quantum algorithm will not be affected. 

Our randomized oracle is defined as follows.   Fix some unknown $a \in \mathbb{Z}_2^n$. Then construct an oracle $\mathcal{O}_a : \ket{x}\ket{y} \mapsto \ket{x}\ket{y + f(x)}$ where $f : \mathbb{Z}_2^n \rightarrow \mathbb{Z}_2^n$ is selected uniformly at random at each call subject to the constraint that $f(x) = f(y)$ if and only if $x = y + a$.  The problem is then to determine $a$.  

Classically, this cannot be done since each query to the oracle results in a random number; however, the quantum algorithm still requires only $\Theta(n)$ queries.

\subsection{An infinity-vs-one separation for the hidden linear structure problem}
Beaudrap, Cleve and Watrous~\cite{beaudrap2001a} introduced the hidden linear structure problem where we are given a blackbox that performs the mapping $\ket{x}\ket{y} \mapsto \ket{x}\ket{\pi(y + sx)}$ where $\pi \in S_q$ and $s \in GF(q)$ for $q = 2^n$.  The problem is to find $s$.  By extending quantum Fourier transforms to $GF(q)$, Beaudrap, Cleve and Watrous~\cite{beaudrap2001a} show that this problem can be solved exactly using a single quantum query but classical algorithms require $\Omega(\sqrt{q})$ queries to determine $s$.   They are able to achieve such a query complexity separation by using a non-standard (but still deterministic) oracle model.  A similar separation ($O(1)$ vs $\Omega(N^{1/4})$) was achieved in the standard oracle model by Aaronson~\cite{aaronson2010}.

We propose the following randomized variant of their oracle problem.  Fix some (unknown) $s \in GF(q)$.  Then define the oracle by $\mathcal{O}_s : \ket{x}\ket{y} \mapsto \ket{x}\ket{\pi(y + sx)}$ where $\pi$ is selected uniformly at random for each query.  The goal is still to determine $s$.  Since the quantum algorithm only uses one query it is unaffected by this change; however, classically the output of the oracle is completely random at each query so we obtain an infinity-vs-one separation.

The three separations shown are examples of a more general phenomenon in which randomness can be used to amplify a modest quantum-vs-classical query separation into an unbounded one.  We discuss this further in Section~\ref{amplification}.

\section{Uselessness for oracles with internal randomness}
\label{uselessness}
We now turn to the general problem of when some number of queries are {\em useless} for solving an oracle problem.  Equivalently we can ask when it is possible to answer an oracle problem with any positive advantage over guessing. 

To define oracle problems, we use a slightly more compact notation than in previous sections.  An oracle ${\pi}$ is defined by a collection of permutations $\pi_{x,r}\in S_M$, where $x\in [N]$ is input by the algorithm, and $r$ is the internal randomness which is distributed according to $R$.  Overloading notation, we say that if the oracle is queried $k$ times, then $\bm{r}=(r_1,\ldots,r_k)$ is distributed according to $R^k$, which may not necessarily be an i.i.d. distribution.  

To describe the problem we want to solve, we follow the notation of Meyer and Pommersheim~\cite{meyer2009a} while adding internal randomness to the oracle.  We are promised that our oracle belongs to a set $C$, which in general may be a strict subset of all functions from $[N]\times\supp(R)$ to $S_M$.  The set $C$ is partitioned into sets $\{C_j\}$, and our goal is to determine which $C_j$ contains ${\pi}$.  By an abuse of notation, we say that $C$ is our oracle problem.  Queries are made to an oracle $\mathcal{O}_\pi$ which acts by $\ket{x}\ket{y} \mapsto \ket{x}\ket{\pi_{x,r}(y)}$.

The oracle problem $C$ is a worst-case problem for which we demand that algorithms work well for all choices of $\pi\in C$.  However, we also consider average-case problems in which $\pi$ is distributed according to a known distribution $\mu$.  The resulting oracle problem is denoted $(C,\mu)$.   


Before stating our own results, we describe the main result of~\cite{meyer2010a}.  In their model there is no internal randomness, so the action of the oracle is simply $\pi_x\in S_M$ for each $x\in [N]$.  If $\bm{x} = (x_1,\ldots,x_k)$ and $\bm{y}=(y_1,\ldots,y_k)$, then define $\pi_\bx(\by) = (\pi_{x_1}(y_1),\ldots,\pi_{x_k}(y_k))$.
Their result may then be stated as follows.

\begin{definition}[Classical uselessness\cite{meyer2010a}]
  $k$ classical queries are useless for the oracle problem $(C,\mu)$ if for all $\bm{x} \in [N]^k$, $\bm{y} \in [M]^k$, $\bz\in[M]^k$ and $j$, $\Pr(\pi \in C_j \mid \pi_\bx(\by)=\bz) = \Pr(\pi \in C_j)$, where $\pi$ is distributed according to $\mu$.
\end{definition}

\begin{definition}[Quantum uselessness\cite{meyer2010a}]
  $k$ quantum queries are useless for the oracle problem $(C,\mu)$ if for any $k$-query quantum algorithm run on any initial state and any POVM measurement $\{M_s\}$ which is made on the output of the algorithm, $\Pr(\pi \in C_j \mid s) = \Pr(\pi \in C_j)$ for all $j$ and $s$, where $\pi$ is distributed according to $\mu$.
\end{definition}

We pause to briefly comment on the connection to unbounded-error query complexity.  Unbounded-error query complexity typically refers to binary problems, i.e. when $C$ is partitioned into $C_0,C_1$ and the goal is to determine which one $\pi$ belongs to with success probability $>1/2$.  In this case, the statement that $k$ (quantum or classical) queries are useless for $(C,\mu)$ (for some $\mu$) is equivalent to the unbounded-error query complexity of $C$ being $> k$.  This is stated precisely and proved in Appendix~\ref{app:unbounded}.

The main result of~\cite{meyer2010a} is the following theorem.
\begin{theorem}[Classical uselessness implies quantum uselessness\cite{meyer2010a}]
  \label{uselessness-theorem}
  For any deterministic oracle problem $(C,\mu)$, if $2 k$ classical queries are useless then $k$ quantum queries  are useless.
\end{theorem}
We will give an alternate proof of this theorem, establish a converse, and generalize it to oracles with internal randomness.

\subsection{Definitions of Classical Uselessness}
In order to characterize uselessness for oracles with internal randomness, we first need to extend the definitions to this case.   As above, we define $\pi_{\bx,\br}(\by) = (\pi_{x_1,r_1}(y_1), \ldots \pi_{x_k,r_k}(y_k))$.  One natural definition of uselessness in this setting is that a classical algorithm ignorant of the oracle's internal randomness should not be able to gain any nontrivial advantage in learning which $C_j$ contains $\pi$.

\begin{definition}[Weak classical uselessness]
  If $(C,\mu)$ is an oracle problem, then $k$ classical queries are weakly useless if for all $\bx\in[N]^k$, $\by,\bz\in[M]^k$ and $j$, $\Pr(\pi \in C_j \mid \pi_{\bx,\br}(\by) =\bz) = \Pr(\pi \in C_j)$, where $\pi$ and $\br$ are distributed according to $\mu$ and $R^k$.
\end{definition}

It is easy to see that if $2 k$ classical queries are weakly useless then $k$ quantum queries need not be useless since Algorithm~\ref{distinguishing-involutions-from-cycles-algorithm} is a counterexample.  A much stronger definition of uselessness would be to allow the classical algorithm to see, or equivalently to choose, the internal random bits used by the oracle.

\begin{definition}[Strong classical uselessness]
  If $(C,\mu)$ is an oracle problem, then  $k$ classical queries are strongly useless if for all $\bm{x} \in [N]^k$, $\bm{y},\bz \in [M]^k$ and all possible values $\bm{r} \in \supp({R^k})$,

  \begin{align}
    \Pr(\pi \in C_j \mid \pi_{\bx,\br}(\by)=\bz) & = \Pr(\pi \in C_j)
  \end{align}
  for all $j$, where $\pi$ is distributed according to $\mu$.
\end{definition}

We will see later that strong classical uselessness for $2 k$ queries is sufficiently powerful to imply quantum uselessness for $k$ queries.  
  However, it is in fact too strong.  A necessary and sufficient condition will come from having each of the $k$ pairs share a seed.

\begin{definition}[Pairwise classical uselessness]
  \label{def:pairwise}
  If $(C,\mu)$ is an oracle problem, then $2 k$ classical queries are pairwise useless if for all $\bm{x}, \bm{x}^\prime \in [N]^k$, $\bm{y}, \bm{y}^\prime,\bz,\bz' \in [M]^k$ and $j$, $\Pr(\pi \in C_j \mid \pi_{\bx,\br}(\by)=\bz, \pi_{\bx',\br}(\by')=\bz') = \Pr(\pi \in C_j)$, where $\pi$ and $\br$ are distributed according to $\mu$ and $R^k$.
\end{definition}

This definition ensures that each pair of query
values $(x_i, x_i^\prime)$ shares the same random seed $r_i$.
We will see later that this corresponds precisely (in the unbounded error setting) to the power of quantum queries, because the density matrix resulting from a quantum query depends on only one random seed, while the different row and column indices interrogate two different choices of $x,y$.

It is important to note that weak classical uselessness and pairwise classical uselessness are not comparable: there exist problems that satisfy weak classical uselessness but not pairwise classical uselessness and vice versa.  Section~\ref{distinguishing-involutions-from-cycles-problem} gives an example where two classical queries are weakly useless but not pairwise useless.  For an example of a problem where two classical queries are not weakly useless but are pairwise useless, let $C$ be the set of all balanced binary functions on $\{0, 1\}$ and let $f$ be chosen uniformly at random from $C$.  Consider the task of determining the function implemented by the oracle that acts for the \nth{i} query by $\ket{x} \mapsto \ket{x \oplus f(r_i)}$ where $r_i$ is the \nth{i} random seed; let $r_1$ be uniformly distributed in $\{0, 1\}$ and let $r_i = 0$ for $i \geq 2$.  Clearly, two classical queries with the random seeds $r_1$ and $r_2$ determine $f$.  However, two classical queries that share the random seed $r_1$ yield no useful information.

It is easy to show that uselessness does not depend on the distribution $\mu(\pi \in C_j)$ over the classes provided the probability of each class is positive.  However, it does depend on the conditional distribution of the oracle within each class.  Consider the problem of determining the parity of a binary function $f : [N] \rightarrow \{0, 1\}$; by tweaking the conditional distribution of $f$ for each parity, we can cause $f(1)$ to be equal to the parity of $f$ with high probability so a single query to $f$ wouldn't be useless.  On the other hand, if the conditional distribution for $f$ were uniform, $N - 1$ classical queries would be useless.



\subsection{Uselessness results}
Our main result in this section is the following equivalence:
\begin{theorem}
  \label{randomized-uselessness-theorem}
For any oracle problem $(C,\mu)$, $k$ quantum queries are useless if and only if 
$2k$ classical queries are pairwise useless.
\end{theorem}

For deterministic oracles, weak, pairwise and strong classical
uselessness are all the same.  In this case,
Theorem~\ref{randomized-uselessness-theorem} can be simplified to the following strengthening of Theorem~\ref{uselessness-theorem}.
\begin{corollary}
  \label{deterministic-uselessness-corollary}
  For any deterministic oracle problem $(C, \mu)$, $k$ quantum queries are useless if and only if 
$2k$ classical queries are useless.
\end{corollary}




\subsection{Encoding oracles in states}
In this section we will prove Theorem~\ref{randomized-uselessness-theorem}. Our strategy will be to show that in the unbounded-error setting, the optimal algorithms make a series of fixed queries and then measure the resulting states.  The key ingredient is to show that oracles can be encoded in states in a way that is perfectly efficient in terms of queries (i.e. one oracle call creates one state, and one state simulates one oracle call), albeit at a cost of producing the output ``I don't know'' most of the time.  We define these encodings first for deterministic oracles.

\begin{definition}\label{def:det-encoding}
Let $\cO_\pi$ be a deterministic permutation oracle that maps $\ket{x,y}\in\bbC^N \ot \bbC^M$  to $\ket{x,\pi_{x}(y)}$.  Then define the {\em encoding} of $\pi$ to be $\ket{\psi_{\pi}} = \frac{1}{\sqrt{NM}} \sum_{x\in [N],y\in [M]} \ket{x}^X \ket{y}^Y \ket{\pi_x(y)}^Z$.  Here $X,Y,Z$ label different registers for notational convenience.
\end{definition}

Clearly one use of $\cO_{\pi}$ allows the creation of one copy of $\ket{\psi_{\pi}}$; simply prepare the state $\frac{1}{\sqrt{NM}}\sum_{x,y} \ket{x}^X\ket{y}^Y\ket{y}^Z$ and apply $\cO_\pi$ to registers $XZ$.  We will see shortly that one copy of $\ket{\psi_{\pi}}$ can in turn simulate one use of $\cO_\pi$, albeit with a very high, but heralded, failure probability.  Before proving this result, we show how \defref{det-encoding} generalizes to oracles with internal randomness.

\begin{definition}
  Let $\cO_\pi$ be an oracle whose action is defined by $\cO_\pi(\ket{x}\bra{x'} \ot \ket{y}\bra{y'} = \E_{r\sim R} \ket{x}\bra{x'} \ot \ket{\pi_{x,r}(y)}\bra{\pi_{x',r}(y')}$.  For each $r$, define the deterministic oracle $\cO_{\pi,r}$ by $\cO_{\pi,r}\ket{x,y} = \ket{x,\pi_{x,r}(y)}$ and define the encoding for fixed $r$ to be $\ket{\psi_{\pi,r}} = \frac{1}{\sqrt{NM}} \sum_{x\in [N], y\in [M]} \ket{x}^X \ket{y}^Y \ket{\pi_{x,r}(y)}^Z$.
\end{definition}

Now we define encodings of oracles with randomness.

\begin{definition}\label{def:rnd-encoding}
  If $\cO_\pi$ is an oracle with internal randomness, then define the {\em encoding} of $\cO_\pi$ to be $\rho_{\pi} = \E_r \psi_{\pi,r}$.
\end{definition}

In this last definition, we use the convention that $\psi =\proj\psi$.  The utility of considering encodings comes from the following operational equivalence.

\begin{theorem}\label{thm:encoding}\ \\[-2ex]
\benum \item One use of $\cO_\pi$ can create one copy of $\rho_{\pi}$.
\item It is possible to consume one copy of $\rho_{\pi}$ and simulate $\cO_\pi$ with success probability $1/NM^2$.  The simulation outputs a classical flag indicating success or failure.
\eenum
In both cases, the run time required is linear in the number of qubits, i.e. $O(\log NM)$.
\end{theorem}
We point out that in the simulation, failure destroys not only the
encoding, but also the state input to the oracle.  Nevertheless, this simulation is enough to distinguish the case when $k$ queries are useless from the case when they are not.

Additionally, \thmref{encoding} is stated implicitly in terms of a distribution $R$.  In the case of $k$ queries correlated according to $R^k$, we have the following variant:
\begin{theorem}\label{thm:cor-encoding}\ \\[-2ex]
\benum
\item $k$ uses of $\cO_\pi$ can create $\rho^k_\pi = \E_{\br\sim R^k} [\psi_{\pi,r_1} \ot \cdots \ot \psi_{\pi,r_k}]$.
\item It is possible to consume $\rho^k_\pi$ and simulate $k$ uses of $\cO_\pi$ with success probability $1/N^kM^{2k}$, again with a flag indicating success or failure.
\eenum
\end{theorem}

As a corollary, for correlated internal randomness in the unbounded-error scenario, we can permit algorithms to make the $k$ oracle calls in any order.  We will prove \thmref{cor-encoding} only, since it subsumes \thmref{encoding}.

\begin{proof}
To create $\rho^k_\pi$, we simply apply $\cO_\pi$ $k$ times to $\left(\frac{1}{\sqrt{NM}}\sum_{x,y} \ket{x}^X\ket{y}^Y\ket{y}^Z\right)^{\ot k}$.

For the second reduction, suppose we are given a copy of $\rho^k_\pi$ and would like to apply $\cO_\pi$ to simulate the \nth{i} query of some algorithm.  If we condition on $\br$, then $\rho^k_\pi$ becomes the state $\psi_{\pi,r_1}\ot\cdots\ot\psi_{\pi,r_k}$.  We will use the \nth{i} component of this state to simulate our query.

Suppose we want to simulate the action of $\cO_{\pi,r_i}$ on the state $\ket{x'}^{X'}\ket{y'}^{Y'}$.  Define
$$A = \sum_{\substack{x\in [N]}}
\ket{x}^X\bra{x,x}^{XX'} \ot 
\frac{1}{\sqrt{M}}\sum_{y\in[M]}\bra{y,y}^{YY'}$$
Since $AA^\dag = \sum_x \proj{x}=I_N$, it follows that $A^\dag A\leq I_{N^2M^2}$ and $\{A,\sqrt{I-A^\dag A}\}$ comprise valid Kraus operators for a quantum operation.  Our simulation will apply this operation, with outcome $A$ labeled success, and $\sqrt{I-A^\dag A}$ labeled failure.

Upon outcome $\sqrt{I-A^\dag A}$, the algorithm declares failure.  If
this occurs at any step of a multi-query algorithm, then the algorithm
should guess $j$ according to the {\em a priori} distribution $\mu$.
Thus, for the purposes of determining whether the algorithm outperforms the
best guessing strategy, it then suffices to consider only the cases
when outcome $A$ occurs.
 
Upon outcome $A$, $\ket{\psi_{\pi,r_i}}^{XYZ}\ket{x'}^{X'}\ket{y'}^{Y'}$ is
mapped to the (unnormalized) state $\frac{1}{\sqrt{NM^2}} \ket{x'}^X
\ket{\pi_{x,r_i}(y')}^Z$.  Since the normalization is independent of the
input, this means that $A$ occurs with probability $1/NM^2$ regardless
of the input state.  Conditioned on this outcome, the resulting map is
precisely the action of $\cO_{\pi,r_i}$.

The overall algorithm succeeds when each of the $k$ queries succeeds.  Since each query succeeds with probability $1/NM^2$, the overall algorithm succeeds with probability $1/N^kM^{2k}$.
\end{proof}

Armed with our notion of encoding, it is straightforward to characterize quantum uselessness.

\begin{corollary}\label{cor:k-useless}
Define $\sigma_j = \E_{\pi \in C_j} \rho^k_\pi$.   
Then $k$ quantum queries are useless if and only if all the  $\sigma_j$ are the same.
\end{corollary}
\begin{proof}
By \thmref{cor-encoding}, any $k$-query algorithm can WLOG create $\rho^k_\pi$, resulting in the state $\sigma_j$ if $\pi$ is drawn randomly from $C_j$.  The algorithm then proceeds to determine which $\sigma_j$ it holds, using no further oracle queries.  If all the $\sigma_j$ are equal, then it can learn nothing about $j$.  Conversely, if some $\sigma_j$ is different from the others, then there is a measurement that will be able to guess $j$ with positive advantage.
\end{proof}


To conclude the proof of Theorem~\ref{randomized-uselessness-theorem}, observe that the quantity on the LHS of \defref{pairwise} is precisely $\tr \left(\ket{\bm{x}}\bra{\bm{x'}} \ot \ket{\by}\bra{\by'} \ot \ket{\bz}\bra{\bz'}\right) \E_{\pi\in C_j} \mu(\pi) \rho^k_\pi = \tr \left(\ket{\bm{x}}\bra{\bm{x'}} \ot \ket{\by}\bra{\by'} \ot \ket{\bz}\bra{\bz'}\right) \sigma_j$ which will be independent of $j$ for all $\bx,\bx',\by,\by',\bz,\bz'$ if and only if all of the $\sigma_j$ are identical.  Combined with Corollary~\ref{cor:k-useless}, this completes the proof of Theorem~\ref{randomized-uselessness-theorem}.

\begin{theorem}
  Suppose that for some oracle problem $k$ classical queries are weakly useless but $k$ quantum queries are not useless.  Then there exists an oracle problem in which this separation holds where the oracle acts by bitwise XOR.
\end{theorem}

\begin{proof}
  Consider an oracle problem $(C,\mu)$ for which $k$ classical queries are weakly useless but $k$ quantum queries are not useless.  The oracle acts by $\mathcal{O} : \ket{x}\ket{y} \mapsto \ket{x}\ket{\pi_{x,r_i}(y)}$ on the \nth{i} call.  We can define a new oracle $\mathcal{O}^\prime : \ket{x}\ket{y}\ket{z} \mapsto \ket{x}\ket{y}\ket{z \oplus \pi_{x,r_i}(y)}$.  Our new oracle $\mathcal{O}^\prime$ can be used to prepare the encoding for $\mathcal{O}$ so $k$ queries to $\mathcal{O}^\prime$ can simulate any quantum algorithm that uses $k$ queries to $\mathcal{O}$.  Classically, $\mathcal{O}^\prime$ can be simulated using $\mathcal{O}$ so we conclude that $k$ classical queries to $\mathcal{O}^\prime$ are weakly useless.
\end{proof}

\section{Amplifying separations}
\label{amplification}
We now leverage our results to obtain a general method of amplifying any separation between classical and quantum uselessness.  Let $(C,\mu)$ be an oracle problem where $C$ is partitioned into $\{C_i\}$ and $r_j$ is the \nth{j} random seed.  For each $\pi \in C$, we have an oracle $\mathcal{O}_\pi$.  Suppose that $k$ classical queries are weakly useless but $k$ quantum queries are not useless.  Let us define the oracle $\mathcal{O}_i : \ket{x_1} \cdots \ket{x_k} \ket{y_1} \cdots \ket{y_k} \mapsto \ket{x_1} \cdots \ket{x_k} \ket{\pi_{x_1, r_1}(y_1)} \cdots \ket{\pi_{x_k, r_k}(y_k)}$ where $\pi$ is selected from $C_i$ according to $\mu$ (this is done independently for each query), $\br$ is distributed according to $R^k$ and a fresh random seed $\bm{r}$ is used for every query to $\mathcal{O}_i$.  Consider the problem of determining $i$ where the oracle $\mathcal{O}_i$ is given with probability $\mu(\pi \in C_i)$.

\begin{theorem}
  \label{thm:gen-classical-uselessness}
  Any number of classical queries to the oracle $\mathcal{O}_i$ is weakly useless for determining $i$.
\end{theorem}

\begin{proof}
  Clearly, a single query to $\mathcal{O}_i$ is equivalent to $k$ queries to the original oracle which are weakly useless by assumption.  We conclude that a single classical query to the new oracle is weakly useless.  We now show that $\ell$ classical queries are weakly useless for any $\ell \geq 1$.  Let $\bx_j \in [N]^k$, $\by_j, \bz_j \in [M]^k$ and let each $\br_j$ be sampled independently from $R^k$ where $1 \leq j \leq \ell$.  We must prove that

  \begin{align}
    \Pr(i \mid \pi^j_{\bx_j, \br_j}(\by_j) = \bz_j, j = 1, \ldots, \ell) & = \Pr(i)
  \end{align}
  where each $\pi^j$ is sampled independently from $C_i$ according to $\mu$.  This condition is equivalent to
  
  \begin{align}
    \Pr(\pi^j_{\bx_j, \br_j}(\by_j) = z_j, j = 1, \ldots, \ell \mid i) & = \Pr(\pi^j_{\bx_j, \br_j}(\by_j) = \bz_j, j = 1, \ldots, \ell)
  \end{align}
  
  Note that by construction, $\Pr(\pi^j_{\bx_j, \br_j}(\by_j) = z_j, j = 1, \ldots, \ell \mid i) = \prod_j \Pr(\pi^j_{\bx_j, \br_j}(\by_j) = \bz_j \mid i)$.  By our assumption that $k$ classical queries to the original oracle are weakly useless, we have that $\Pr(i \mid \pi_{\bx_j, \br_j}(\by_j) = \bz_j) = \Pr(i)$ or equivalently $\Pr(\pi_{\bx_j, \br_j}(\by_j) = \bz_j \mid i) = \Pr(\pi_{\bx_j, \br_j}(\by_j) = \bz_j)$.  Therefore,
  
  \begin{align}
    \Pr(\pi^j_{\bx_j, \br_j}(\by_j) = \bz_j, j = 1, \ldots, \ell) & = \sum_i \Pr(\pi^j_{\bx_j, \br_j}(\by_j) = z_j, j = 1, \ldots, \ell \mid i) \Pr(i) \\
                                                               {} & = \sum_i \left(\prod_j \Pr(\pi^j_{\bx_j, \br_j}(\by_j) = z_j \mid i)\right) \Pr(i) \\
                                                               {} & = \sum_i \left(\prod_j \Pr(\pi^j_{\bx_j, \br_j}(\by_j) = z_j)\right) \Pr(i) \\
                                                               {} & = \prod_j \Pr(\pi^j_{\bx_j, \br_j}(\by_j) = z_j) \\
                                                               {} & = \prod_j \Pr(\pi^j_{\bx_j, \br_j}(\by_j) = z_j \mid i) \\
                                                               {} & = \Pr(\pi^j_{\bx_j, \br_j}(\by_j) = z_j, j = 1, \ldots, \ell \mid i)
  \end{align}
  which is the desired result.
\end{proof}

We conclude that no matter how many classical queries are made to $\mathcal{O}_i$, no information is obtained about $i$.  On the other hand, we have the following result:

\begin{theorem}
  \label{thm:gen-quantum-uselessness}
  A single quantum query to $\mathcal{O}_i$ is \emph{not} useless for determining $i$.
\end{theorem}

\begin{proof}
  One can use a single quantum query to $\mathcal{O}_i$ to construct the state $\rho_\pi^k$ as described in Theorem~\ref{thm:cor-encoding}.  Applying Theorem~\ref{thm:cor-encoding}, this state may be used to guess $i$ with higher probability than random guessing since $k$ quantum queries are not useless.
\end{proof}

Thus, we have constructed an infinity-vs-one separation in unbounded-error classical and quantum query complexities from an arbitrary initial separation.  One can also construct an infinity-vs-one separation in the bounded-error setting from an arbitrary separation in the unbounded setting; the construction is straightforward and we defer the details to Appendix~\ref{sec:bounded-amp}.

\section*{Acknowledgments}
We thank an anonymous reviewer for suggesting the infinity-vs-one separation for the problem in \cite{beaudrap2001a}.  DJR was funded by NSF grant CCF-0916400 and by the DoD through an NDSEG fellowship.   
AWH was funded by NSF grants CCF-0916400 and CCF-1111382, ARO contract
W911NF-12-1-0486 and IARPA QCS program (D11PC20167).    Much of this work was done while AWH worked at the University of Washington.
The views and conclusions contained herein are those of the authors
and should not be interpreted as necessarily representing the official
policies or endorsements, either expressed or implied, of IARPA,
DoI/NBC, or the U.S. Government.

\newpage
\appendix

\newpage
\section{Alternate proofs of uselessness}
In this appendix, we present alternate proofs of various uselessness theorems.  These proofs do not rely on the idea of encoding oracles into states, but instead give direct arguments, so they are more self-contained, although also longer.  First we prove that pairwise classical uselessness implies quantum uselessness.

\begin{proof}
  The proof is an extension of the technique used by Meyer and Pommersheim.  Suppose that $2 k$ classical queries are pairwise useless.  Consider an oracle $\pi$ that acts by $\mathcal{O}_\pi^i : \ket{x, y, z} \mapsto \ket{x, \pi_{x, r_i}, z}$ for the \nth{i} query.  Note that, as before, the $r_i$ variables may obey an arbitrary joint distribution so different queries are not necessarily independent.  Consider an arbitrary $k$-query quantum algorithm with initial state $\rho_0$ and POVM $\{M_s\}$.  For the \nth{i} query, the algorithm queries the oracle and then applies an arbitrary unitary transformation $U_i$.  This yields the final state

  \begin{align}
    \rho_\pi & = U_k \mathcal{O}_\pi^k \ldots U_1 \mathcal{O}_\pi^1 \rho_0 {\mathcal{O}_\pi^1}^\dagger U_1^\dagger \ldots {\mathcal{O}_\pi^k}^\dagger U_k^\dagger
  \end{align}
  Let us fix the random seed used for the \nth{i} query as $r_i$.  The final state is then

    \begin{align}
    \rho_{\pi, \bm{r}} & = U_k P_{r_k} \ldots U_1 P_{r_1} \rho_0 P_{r_1}^\dagger U_1^\dagger \ldots P_{r_k}^\dagger U_k^\dagger
  \end{align}
  where $P_{r_i}$ denotes the permutative action $\ket{x, y, z} \mapsto \ket{x, \pi_{x, r_i}(y), z}$ of the oracle when the random seed is fixed to $r_i$.  Let $A$ be a matrix, $L = (x, y, z)$ and $L^\prime = (x^\prime, y^\prime, z^\prime)$.  Then

  \begin{align}
    \left(P_{r_i} A P_{r_i}^\dagger\right)_{L, L^\prime} & = \bra{x, \pi_{x, r_i}^{-1}(y), z} A \ket{x^\prime, \pi_{x^\prime, r_i}^{-1}(y^\prime), z^\prime} \\
                                                      {} & = A_{\pi_{\cdot, r_i}(L), \pi_{\cdot, r_i}(L^\prime)}
  \end{align}
  where $\pi_{\cdot, r_i}(L) = (x, \pi_{x, r_i}^{-1}(y), z)$.  Then the state after the \nth{i + 1} query (for the fixed values $\bm{r}$ of the seeds) is

  \begin{align}
    \rho_{i + 1, \bm{r}} & = U_{i + 1} P_{r_{i + 1}} \rho_{i, \bm{r}} P_{r_{i + 1}}^\dagger U_{i + 1}^\dagger
  \end{align}
  so that the matrix elements are

  \begin{align}
    (\rho_{i + 1, \bm{r}})_{L, L^\prime} & = \sum_{K, K^\prime} (U_{i + 1})_{L, K} (\rho_{i, \bm{r}})_{\pi_{\cdot, r_{i + 1}}(K), \pi_{\cdot, r_{i + 1}}(K^\prime)}  (U_{i + 1}^\dagger)_{K^\prime, L^\prime}
  \end{align}
  This value is a function of $L$, $L^\prime$, $\pi_{\cdot, r_{i + 1}}(K)$ and $\pi_{\cdot, r_{i + 1}}(K^\prime)$.  Therefore, the final state $\rho_{\pi, \bm{r}} = \rho_{k, \bm{r}}$ may be written as

  \begin{align}
    \rho_{\pi, \bm{r}} & = \sum_{I} Q_I(\pi_{\bm{x}, \bm{r}}(\bm{y}), \pi_{\bm{x}^\prime, \bm{r}}(\bm{y}^\prime))
  \end{align}
  where $I = (L_1, \ldots, L_k, L_1^\prime \ldots, L_k^\prime)$.  Let $\E_{\pi \mid \pi \in C_j}$ denote the expectation over $\pi$ according to the distribution $\Pr(\pi \mid \pi \in C_j)$.  Then for any $j$,

  \begin{align}
    \E_{\pi \mid \pi \in C_j} \rho_{\pi, \bm{r}} & = \sum_I \E_{\pi \mid \pi \in C_j} Q_I(\pi_{\bm{x}, \bm{r}}(\bm{y}), \pi_{\bm{x}^\prime, \bm{r}}(\bm{y}^\prime)) \\
                                              {} & = \sum_I \sum_{\bm{w}, \bm{w^\prime}} \E_{\pi \mid \pi \in C_j} Q_I(\pi_{\bm{x}, \bm{r}}(\bm{y}), \pi_{\bm{x}^\prime, \bm{r}}(\bm{y}^\prime)) [\pi_{\bm{x}, \bm{r}}(\bm{y}) = \bm{w}, \pi_{\bm{x}^\prime, \bm{r}}(\bm{y}^\prime) = \bm{w}^\prime]) \\
                                                   & \text{ where } \bm{w} = (w_1, \ldots, w_k) \text{ and } \bm{w}^\prime = (w_1^\prime, \ldots, w_k^\prime) \nonumber \\
                                              {} & = \sum_I \sum_{\bm{w}, \bm{w^\prime}} Q_I(\bm{w}, \bm{w}^\prime) \E_{\pi \mid \pi \in C_j} [\pi_{\bm{x}, \bm{r}}(\bm{y}) = \bm{w}, \pi_{\bm{x}^\prime, \bm{r}}(\bm{y}^\prime) = \bm{w}^\prime]) \\
                                              {} & = \sum_I \sum_{\bm{w}, \bm{w^\prime}} Q_I(\bm{w}, \bm{w}^\prime) \Pr(\pi_{\bm{x}, \bm{r}}(\bm{y}) = \bm{w}, \pi_{\bm{x}^\prime, \bm{r}}(\bm{y}^\prime) = \bm{w}^\prime \mid \pi \in C_j) \\
  \end{align}
  Taking the expectation over the random seeds $\bm{r}$,

  \begin{align}
    \E_{\pi \mid \pi \in C_j} \E_{\bm{r}} \rho_{\pi, \bm{r}} {} & = \sum_I \sum_{\bm{w}, \bm{w^\prime}} Q_I(\bm{w}, \bm{w}^\prime) \E_{\bm{r}} \Pr(\pi_{\bm{x}, \bm{r}}(\bm{y}) = \bm{w}, \pi_{\bm{x}^\prime, \bm{r}}(\bm{y}^\prime) = \bm{w}^\prime \mid \pi \in C_j) \\
                                                             {} & = \sum_I \sum_{\bm{w}, \bm{w^\prime}} Q_I(\bm{w}, \bm{w}^\prime) \Pr(\pi_{\bm{x}, \bm{r}}(\bm{y}) = \bm{w}, \pi_{\bm{x}^\prime, \bm{r}}(\bm{y}^\prime) = \bm{w}^\prime \mid \pi \in C_j) \\
                                                             {} & = \sum_I \sum_{\bm{w}, \bm{w^\prime}} Q_I(\bm{w}, \bm{w}^\prime) \Pr(\pi \in C_j \mid \pi_{\bm{x}, \bm{r}}(\bm{y}) = \bm{w}, \pi_{\bm{x}^\prime, \bm{r}}(\bm{y}^\prime) = \bm{w}^\prime) \\
                                                                  & {} \cdot \frac{\Pr(\pi_{\bm{x}, \bm{r}}(\bm{y}) = \bm{w}, \pi_{\bm{x}^\prime, \bm{r}}(\bm{y}^\prime) = \bm{w}^\prime)}{\Pr(\pi \in C_j)} \\
                                                             {} & = \sum_I \sum_{\bm{w}, \bm{w^\prime}} Q_I(\bm{w}, \bm{w}^\prime) \Pr(\pi_{\bm{x}, \bm{r}}(\bm{y}) = \bm{w}, \pi_{\bm{x}^\prime, \bm{r}}(\bm{y}^\prime) = \bm{w}^\prime) \\
                                                                  & \text{by pairwise classical uselessness} \nonumber \\
                                                             {} & = \E_\pi \E_r \sum_I \sum_{\bm{w}, \bm{w^\prime}} Q_I(\bm{w}, \bm{w}^\prime) [\pi_{\bm{x}, \bm{r}}(\bm{y}) = \bm{w}, \pi_{\bm{x}^\prime, \bm{r}}(\bm{y}^\prime) = \bm{w}^\prime] \\
                                                             {} & = \E_\pi \E_r \sum_I Q_I(\pi_{\bm{x}, \bm{r}}(\bm{y}), \pi_{\bm{x}^\prime, \bm{r}}(\bm{y}^\prime)) \\
                                                             {} & = \E_\pi \E_{\bm{r}} \rho_{\pi, \bm{r}} \\
  \end{align}
  Defining $\rho_\pi = \E_{\bm{r}} \rho_{\pi, \bm{r}}$, this may be written as

  \begin{align}
    \E_{\pi \mid \pi \in C_j} \rho_\pi & = \E_\pi \rho_\pi
  \end{align}
  Note that for a random $\pi \in C$, the state after running the algorithm is $\E_\pi \rho_\pi$ and for a random $\pi \in C_j$ the state is $\E_{\pi \mid \pi \in C_j}\rho_\pi$.  Now, consider the probability that $\pi \in C_j$ given the measurement outcome $s$.  We have

  \begin{align}
    \Pr(\pi \in C_j \mid s) & = \frac{\Pr(s \mid \pi \in C_j) \Pr(\pi \in C_j)}{\Pr(s)} \\
                         {} & = \frac{\tr M_s \E_{\pi \mid \pi \in C_j} \rho_f}{\tr M_s \E_\pi \rho_\pi} \Pr(\pi \in C_j) \\
                         {} & = \Pr(\pi \in C_j)
  \end{align}
  as claimed.
\end{proof}

Next, we prove that quantum uselessness implies classical uselessness, but in the special case of standard oracles that act via XOR but with internal randomness.  Specifically, consider an oracle that acts by $\mathcal{O}_f^i : \ket{x, y, z} \mapsto \ket{x, y \oplus f(x, r_i), z}$ for the \nth{i} query.  As before, we allow the $r_i$ variables to be drawn from an arbitrary joint distribution.  

\begin{proof}
  Suppose that $k$ quantum queries are useless.  This means that for any POVM $\{M_s\}$ and quantum algorithm run on any initial state $\rho_0$, $\Pr(f \in C_j \mid s) = \Pr(f \in C_j)$ for all $j$.  Since $\Pr(f \in C_j \mid s) = \frac{\Pr(s \mid f \in C_j) \Pr(f \in C_j)}{\Pr(s)}$, this implies that

  \begin{align}
    \Pr(s \mid f \in C_j) & = \Pr(s)
  \end{align}
  for all $j$.  Let us choose the initial state

  \begin{align}
    \rho_0 & = \left(\frac{1}{N} \sum_{x, x^\prime} \ket{x} \bra{x^\prime} \otimes \ket{0} \bra{0}\right)^{\otimes k}
  \end{align}
  and the algorithm defined by the unitary operator $\bigotimes_{i = 1}^k \mathcal{O}_f^i$.  The result of running the algorithm assuming a particular function $f$ and fixed seeds $\bm{r}$ is then

  \begin{align}
    \rho_{f, \bm{r}} & = \left(\bigotimes_{i = 1}^k \mathcal{O}_f^i\right) \rho_0 \left(\bigotimes_{i = 1}^k \mathcal{O}_f^i\right)^\dagger \\
                  {} & = \frac{1}{N^k} \bigotimes_{i = 1}^k \sum_{x, x^\prime} \ket{x, f(x, r_i)} \bra{x^\prime, f(x^\prime, r_i)}
  \end{align}
  For a particular function $f$, the state after running the algorithm is

  \begin{align}
    \rho_f & = \E_{\bm{r}} \rho_{f, \bm{r}} \\
        {} & = \frac{1}{N^k} \E_{\bm{r}} \bigotimes_{i = 1}^k \sum_{x, x^\prime} \ket{x, f(x, r_i)} \bra{x^\prime, f(x^\prime, r_i)}
  \end{align}
  Now

  \begin{align}
    \Pr(s) & = \E_f \Pr(s \mid f) \\
        {} & = \tr M_s \E_f \rho_f \\
        {} & = \tr M_s \rho_C
  \end{align}
  Similarly,

  \begin{align}
    \Pr(s \mid f \in C_j) & = \E_{f \mid f \in C_j} \Pr(s \mid f) \\
                       {} & = \tr M_s \E_{f \mid f \in C_j} \rho_f \\
                       {} & = \tr M_s \rho_{C_j}
  \end{align}
  Since $\Pr(s \mid f \in C_j) = \Pr(f \in C_j)$, this implies that

  \begin{align}
    \tr M_s (\rho_{C_j} - \rho_C) & = 0
  \end{align}
  for \emph{all} POVMs $\{M_s\}$ which means that

  \begin{align}
    \rho_{C_j} & = \rho_C \\
    \frac{1}{N^k} \E_{\bm{r}} \E_{f \mid f \in C_j} \bigotimes_{i = 1}^k \sum_{x, x^\prime} \ket{x, f(x, r_i)} \bra{x^\prime, f(x^\prime, r_i)} & = \frac{1}{N^k} \E_{\bm{r}} \E_f \bigotimes_{i = 1}^k \sum_{x, x^\prime} \ket{x, f(x, r_i)} \bra{x^\prime, f(x^\prime, r_i)}
  \end{align}
  Equating the $((x_1, y_1, \ldots, x_k, y_k), (x_1^\prime, y_1^\prime, \ldots, x_k^\prime, y_k^\prime))$ elements of these matrices, we have that

  \begin{align}
    \E_{\bm{r}} \E_{f \mid f \in C_j} [f(\bm{x}, \bm{r}) = \bm{y}, f(\bm{x}^\prime, \bm{r}) = \bm{y}^\prime] & = \E_{\bm{r}} \E_f [f(\bm{x}, \bm{r}) = \bm{y}, f(\bm{x}^\prime, \bm{r}) = \bm{y}^\prime] \\
    \E_{\bm{r}} \Pr(f(\bm{x}, \bm{r}) = \bm{y}, f(\bm{x}^\prime, \bm{r}) = \bm{y}^\prime \mid f \in C_j) & = \E_{\bm{r}} \Pr(f(\bm{x}, \bm{r}) = \bm{y}, f(\bm{x}^\prime, \bm{r}) = \bm{y}^\prime) \\
    \Pr(f(\bm{x}, \bm{r}) = \bm{y}, f(\bm{x}^\prime, \bm{r}) = \bm{y}^\prime \mid f \in C_j) & = \Pr(f(\bm{x}, \bm{r}) = \bm{y}, f(\bm{x}^\prime, \bm{r}) = \bm{y}^\prime)
  \end{align}
  Applying Bayes' rule, we have

  \begin{align}
    \Pr(f \in C_j \mid f(\bm{x}, \bm{r}) = \bm{y}, f(\bm{x}^\prime, \bm{r}) = \bm{y}^\prime) & = \frac{\Pr(f(\bm{x}, \bm{r}) = \bm{y}, f(\bm{x}^\prime, \bm{r}) = \bm{y}^\prime \mid f \in C_j) \Pr(f \in C_j)}{\Pr(f(\bm{x}, \bm{r}) = \bm{y}, f(\bm{x}^\prime, \bm{r}) = \bm{y}^\prime)} \\
                                                                                          {} & = \Pr(f \in C_j)
  \end{align}
  which is precisely the definition of pairwise classical uselessness in the case of oracles that act by XOR.
\end{proof}

Combining this with Theorem~\ref{randomized-uselessness-theorem}, we have the following result

\begin{corollary}
  For any oracle problem $(C,\mu)$ in the standard model with internal randomness, $k$ quantum queries are useless if and only if $2 k$ classical queries are pairwise useless
\end{corollary}

Since pairwise classical uselessness is equivalent to classical uselessness when $f$ is deterministic, we have the following corollary.

\begin{corollary}
  If $k$ quantum queries are useless for an oracle problem $(C,\mu)$ in the standard model, then $2 k$ classical queries are useless.
\end{corollary}

\newpage
\section{Bounded-error infinity-vs-one separations}
\label{sec:bounded-amp}
We now show how to obtain an infinity-vs-one separation in the bounded-error regime from an arbitrary separation between the classical and quantum uselessness.  Consider the oracle $\mathcal{O}_i$ as defined above.  By Theorem~\ref{thm:gen-quantum-uselessness}, there exists a single-query quantum algorithm $A$, a POVM $\{M_s\}$ and an $i^\prime$ such that for some $s$, $\Pr(i = i^\prime \mid s) > \Pr(i = i^\prime)$.  Equivalently, 

\begin{align}
                              \Pr(s \mid i = i^\prime) & > \Pr(s) \\
      \Pr(s \mid i = i^\prime) (1 - \Pr(i = i^\prime)) & > \Pr(s \mid i \not= i^\prime) \Pr(i \not= i^\prime) \\
                              \Pr(s \mid i = i^\prime) & > \Pr(s \mid i \not= i^\prime) \\
                              \Pr(s \mid i = i^\prime) & = \Pr(s \mid i \not= i^\prime) + \epsilon
\end{align}
for some $\epsilon > 0$.  Consider the problem of deciding if $i = i^\prime$ by querying $\mathcal{O}_i$.  By running $A$ some large number of times $T$ and using majority voting and Chernoff bounds, we may decide if $i = i^\prime$ with bounded error.  Although $T$ may be quite large, the gap is large since it is a separation between an infinite number of classical queries and a finite number of classical queries.

\begin{corollary}
  The bounded-error quantum query complexity of deciding if $i = i^\prime$ using $\mathcal{O}_i$ is finite.
\end{corollary}

By Theorem~\ref{thm:gen-classical-uselessness}, $\Pr(i \mid \pi^j_{\bx_j, \br_j}(\by_j) = z_j, j = 1, \ldots) = \Pr(i)$ for all $\ell \geq 1$ and $\bx_j \in [N]^k$, $\by_j, \bz_j \in [M]^k$.  Thus, $\Pr(i = i^\prime \mid \pi^j_{\bx_j, \br_j}(\by_j) = z_j, j = 1, \ldots) = \Pr(i = i^\prime)$ and $\Pr(i \not= i^\prime \mid \pi^j_{\bx_j, \br_j}(\by_j) = z_j, j = 1, \ldots) = \Pr(i \not= i^\prime)$ so $\ell$ queries are weakly useless for deciding if $i = i^\prime$.

\begin{corollary}
  Any number of classical queries to the oracle $\mathcal{O}_i$ is weakly useless for deciding if $i = i^\prime$; thus no classical algorithm can decide if $i = i^\prime$ with unbounded error no matter how many queries are made.
\end{corollary}

We can construct a new oracle $\mathcal{O}_i^\prime$ that simulates $T$ queries to $\mathcal{O}_i$ using an independent random seed for each query.  From this we obtain the following.

\begin{corollary}
  The bounded-error quantum query complexity of deciding if $i = i^\prime$ using $\mathcal{O}'_i$ is $1$.
\end{corollary}

\begin{corollary}
  Any number of classical queries to the oracle $\mathcal{O}_i^\prime$ is weakly useless for deciding if $i = i^\prime$; thus no classical algorithm can decide if $i = i^\prime$ with unbounded error no matter how many queries are made.
\end{corollary}

Thus, we have constructed an infinity-vs-one separation between the bounded-error quantum query complexity and the unbounded-error classical query complexity from an arbitrary initial separation.  This comes at the price of large inputs for the constructed oracle.

\section{Relation between uselessness and unbounded query complexity}\label{app:unbounded}

In this section, we define {\em binary} oracle problems to be those where our goal is to output a single bit, or equivalently, where $C$ is partitioned into only two sets $C_0,C_1$, and our goal is to determine whether $\pi\in C_0$ or $\pi\in C_1$.
\begin{proposition} 
Let $C$ be a binary oracle problem.  Then the unbounded quantum (resp. classical) query complexity of $C$ is $>k$ if and only if there exists a distribution $\mu$ with $\mu(C_0)=\mu(C_1)=1/2$ such that $k$ quantum (resp. classical) queries are useless for $(C,\mu)$.  
\end{proposition}
Equivalently we could demand that $0<\mu(C_0)<1$ because reweighting 0-inputs and 1-inputs does not affect the uselessness properties of a distribution.  (The same does {\em not} hold for changing the probabilities within the class of 0-inputs or 1-inputs.)  However, we need to avoid the trivial case in which a distribution is useless because the answer is already known perfectly from the prior distribution $\mu$.

\begin{proof}
The ``if'' direction is easy.  If such a $\mu$ exists, then by the definition of uselessness, no algorithm can achieve success probability $>1/2$ with $\leq k$ queries.

For the converse, we use Yao's minimax principle, which states that there exists a distribution $\mu$ for which no $k$-query algorithm can achieve success probability $>1/2$.  Since it is always possible to achieve success probaiblity $\max(\mu(C^{-1}(0)), \mu(C^{-1}(1)))$ by guessing, we must also have $\mu(C^{-1}(0))=\mu(C^{-1}(1))=1/2$.
\end{proof}

A natural generalization of unbounded-error query complexity to non-binary problems would be to define success as guessing the right answer with probability $>\max_j \mu(C_j)$.  In this case, uselessness is now a strictly stronger statement whenever $\mu$ is such that $\mu(C_j)$ is not the same for each $j$.   To see this, let $\nu$ be the distribution over $\pi$ obtained after making some number of queries.  Uselessness states that $\mu(C_j) = \nu(C_j)$ for each $j$, whereas unbounded-error query complexity depends only on whether $\max_j \mu(C_j) = \max_j \nu(C_j)$.

\bibliographystyle{initials}
\bibliography{computer-science-references,quantum-computing-references}

\end{document}